\documentclass{article}
\usepackage[utf8]{inputenc}

\usepackage[T1]{fontenc}
\usepackage{soul, comment} 
\usepackage{phfqit}
\usepackage{proba} 
\usepackage[sort,nocompress]{cite}
\usepackage{amssymb}
\usepackage{amsmath}
\usepackage{amsthm}
\usepackage{amsfonts}
\usepackage{mathtools}
\usepackage{xspace}
\usepackage{bm}
\usepackage{nicefrac}
\usepackage{commath}
\usepackage{bbm}
\usepackage{boxedminipage}
\usepackage{xparse}
\usepackage{xcolor}
\usepackage{float}
\usepackage{multirow}
\usepackage{makecell}
\usepackage{graphicx}
\usepackage{caption}
\usepackage{subcaption}
\usepackage{ifthen}
\usepackage{thm-restate}
\usepackage{algorithm}
\usepackage[noend]{algpseudocode}
\usepackage{colortbl} 
\usepackage{fancyhdr}   
\usepackage{hyperref}
    \hypersetup{ colorlinks=true, linkcolor=blue, citecolor=magenta, urlcolor=blue,}
\usepackage{fullpage}
\usepackage[utf8]{inputenc}
\usepackage{environ}
\usepackage{mdframed}
\usepackage{cleveref}
\usepackage[short]{optidef} 
\usepackage{tikz}
\usetikzlibrary{shapes}
\usetikzlibrary{positioning}
\usetikzlibrary{fit}
\usetikzlibrary{graphs,graphs.standard}
\usepackage{enumitem}

\newtheorem{proposition}{Proposition}[section]
\newtheorem{lemma}{Lemma}[section]
\newtheorem{remark}{Remark}[section]
\newtheorem{theorem}{Theorem}[section]

\newtheorem{claim}{Claim}[section]

\usepackage{enumitem}

\NewEnviron{problem}[1]{%
	\begin{center}\fbox{\parbox{6in}{%
				{\centering\scshape #1\par}%
				\parskip=1ex
				\everypar{\hangindent=1em}%
				\BODY
}}\end{center}}

\newcommand{\cC}{\ensuremath{{\mathcal C}}\xspace}

\newcommand{\cT}{\ensuremath{{\mathcal T}}\xspace}

\newcommand{\stpqFGC}{$(p,q)$-FlexST\xspace}
\newcommand{\pqFGC}{$(p,q)$-FGC\xspace}
\newcommand{\pqFGCAug}{$(p,q)$-FGC-Aug\xspace}

\newcommand{\pqFlexC}{$(p,q)$-Flex-Connected\xspace}
\newcommand{\pqminusoneFlexC}{$(p,q-1)$-Flex-Connected\xspace}
\newcommand{\rootedCAug}{Rooted-Conn-Aug\xspace}

\title{A $(p+q)^{O(pq)}$-approximation for $(p, q)$-Flexible Graph Connectivity}

\author{
Karthekeyan Chandrasekaran\thanks{University of Illinois, Urbana-Champaign, USA, Email: \{karthe, ryjiang2, kk17\}@illinois.edu. Karthekeyan is supported in part by NSF grant CCF-2402667. 
}
\and Raymond Jiang\footnotemark[1]
\and Krishna Kalathur\footnotemark[2]
}
\date{}
\begin{document}
\maketitle
\begin{abstract}
    In the $(p,q)$-Flexible Graph Connectivity problem, the input consists of non-negative integers $p$ and $q$ and a graph $G=(V, E)$ whose edges are classified into safe and unsafe edges with non-negative edge costs. 
    A subgraph $H$ of $G$ is \pqFlexC if every non-empty proper subset $B$ of vertices has either at least $p$ safe edges or at least $p+q$ total edges crossing it. 
    The goal is to find a minimum cost subset $F\subseteq E$ of edges such that the subgraph $(V, F)$ is \pqFlexC. 
    We give a $(p+q)^{O(pq)}$-approximation for this problem, which in particular implies a constant approximation for every fixed constants $p$ and $q$. 
    We achieve this by designing a $(p+q)^{O(pq)}$-approximation for the 
    augmentation problem of finding a minimum cost subset of edges to add to make a \pqminusoneFlexC graph into a \pqFlexC graph. Underlying the augmentation algorithm is a  structural result showing that all deficient cuts can be represented by min rooted-cuts in a $(p+q)^{pq}$-sized collection of digraphs. This structural result was discovered by ChatGPT Astra. 
\end{abstract}


\section{Introduction}\label{sec:intro}
We consider $(p, q)$-Flexible Graph Connectivity problem. We set up the notation to describe the problem. Let $G=(V, E)$  be a graph whose edges are partitioned into ``safe'' and ``unsafe'' edges. 
For a subset $B\subseteq V$ of vertices and a subset $F\subseteq E$, we write $\delta_F(B)$ to be the subset of edges of $F$ crossing $B$, i.e., the subset of edges of $F$ with exactly one end-vertex in $B$. 
For positive integers $p, q\ge 0$, the graph $G$ is \pqFlexC if every non-empty proper subset $B\subseteq V$ has either at least $p$ safe edges crossing it or at least $p+q$ total edges crossing it. 
We now define the $(p, q)$-Flexible Graph Connectivity problem, abbreviated \pqFGC. 

\begin{center}
\fbox{%
\begin{minipage}{0.95\textwidth}
\textbf{\pqFGC:}

\textbf{Given:} Non-negative integers $p$ and $q$, undirected graph $G=(V, E=S\uplus U)$ with costs $c: E\rightarrow \R_{\ge 0}$.

\textbf{Goal:} Find a subset $F\subseteq E$ with minimum $\sum_{e\in F}c_e$ such that $(V, F)$ is \pqFlexC. 

\end{minipage}}
\end{center}

$(p=1, q=1)$-FGC was introduced by Adjiashvili, Hommelsheim, and M\"{u}hlenthaler \cite{AHM2022} to model non-uniform failure models. For general $p$ and $q$, \pqFGC was introduced by Boyd, Cheriyan, Haddadan, and Ibrahimpur \cite{BCHI24}. 
Subsequent to the introduction of \pqFGC, several other models of network design for non-uniform failure models have been introduced and studied, but we do not describe those models here---see \cite{AHMS22, BCHI24, HJS24, CJ25, HLMZ25}. \pqFGC generalizes the $k$-edge-connected spanning subgraph problem, which is APX-hard \cite{Fer98} and admits a $2$-approximation \cite{Jain2001}. 
It is known that feasibility testing of a given instance of \pqFGC can be done in polynomial time via enumeration of approximate min-cuts \cite{BCHI24}. 
Adjiashvili, Hommelsheim, and M\"{u}hlenthaler showed constant-factor approximation for $(1, 1)$-FGC. 
After that initial result, several works \cite{AHM2022, BCGI23, BCHI24, HJS24, CJ25, BCKS25, Nut25, Ban25, Nut25-tight, HLMZ25} have designed constant-factor approximations for various regimes of $p$ and $q$. 
Ibrahimpur and V\'{e}gh \cite{IV26-j} recently designed an $O(\log{n})$-approximation for all $p$ and $q$. 
Beyond approximation guarantees, \pqFGC has reinvigorated the field of network design through novel concepts, problems, and results for classic problems.  
It drove extensions of the primal-dual method beyond uncrossable families through the introduction of pliable families \cite{BCGI23, Ban25, Nut25, Nut25-tight}, motivated the study of small-cut covering problems \cite{BCGI23, BCKS25}, and led to improved approximations for the capacitated $k$-edge-connected spanning subgraph problem \cite{BCGI23, Ban25, BCKS25}. 

An intriguing central question concerning \pqFGC is whether, for all $p$ and $q$, there exists an $\alpha(p,q)$-approximation for some function $\alpha(p, q)$ that depends only on $p$ and $q$. Such an approximation would imply a constant-factor approximation for every fixed constant $p$ and $q$. We resolve this question affirmatively in this work. The key ideas behind this work were discovered by ChatGPT Astra. We simplify the ideas and present them from our viewpoint. 

As mentioned above, feasibility can be verified in polynomial time. 
We henceforth assume that all input instances are feasible. To state our approximation factor, we use $L(p, 0):=1$, 
\begin{align*}
    L(p, q)&:=\prod_{\ell=p}^{p+q-1}\binom{\ell}{p-1}\ \forall\ q\ge 1, \text{ and}\\
    \alpha(p, q) &:=\sum_{j=0}^q jL(p, j)\ \forall\ p, q\ge 0. 
\end{align*}
We observe that $L(p, q)\le(p+q)^{(p-1)q}$ and $\alpha(p, q)=(p+q)^{O(pq)}$. We show the following result. 

\begin{theorem}\label{thm:pqFGC}
    There exists a randomized polynomial-time algorithm for \pqFGC, that returns a feasible solution $F$ with $c(F)\le O(\alpha(p, q))OPT$ with constant probability, where OPT is the optimum solution cost. 
\end{theorem}

To prove Theorem \ref{thm:pqFGC}, we address the following augmentation problem: 

\begin{center}
\fbox{%
\begin{minipage}{0.95\textwidth}
\textbf{\pqFGCAug:}

\textbf{Given:} Non-negative integers $p$ and $q$, 
undirected graph $H=(V, E=S\uplus U)$ where edges in $S$ are safe and edges in $U$ are unsafe such that $H$ is \pqminusoneFlexC, 
a set $N$ of undirected edges over vertex set $V$ along with a classification into safe and unsafe edges, and non-negative costs $c: N\rightarrow \R_{\ge 0}$.

\textbf{Goal:} Find a subset $F\subseteq E$ with minimum $\sum_{e\in F}c_e$ such that $(V, E\cup F)$ is \pqFlexC. 

\end{minipage}}
\end{center}

We recall that feasibility testing can be done in polynomial time via enumeration of approximate min-cuts---e.g., see \cite{BCHI24}. We will henceforth assume that all input instances for the augmentation problem are also feasible. 
We clarify a subtle nuance in the augmentation problem that renders the precise classification of the edges of $N$ into safe and unsafe edges irrelevant. To see this, it is helpful to reformulate \pqFGCAug as a cut-covering problem. We break symmetry by fixing a root vertex $r\in V$ and focus on deficient-cuts that do not contain $r$ as defined below: 
\[
\cT:=\left\{\emptyset\neq B\subseteq V-r: |\delta_E(B)|=p+q-1, |\delta_{S}(B)|\le p-1\right\}. 
\]
It is known that for a subset $F\subseteq N$, the graph $(V, E\cup F)$ is \pqFlexC if and only if $|\delta_F(B)|\ge 1$ for every $B\in \cT$ (e.g., see \cite{BCHI24, CJ25}). 
Consequently, the precise classification of the edges of $N$ into safe and unsafe edges is irrelevant. We show the following result for \pqFGCAug. 

\begin{theorem}\label{thm:pqFGCAug}
    There exists a randomized polynomial-time algorithm for \pqFGCAug, that returns a feasible augmentation $F$ with $c(F)= O(qL(p, q))\text{OPT}_{\text{Aug}}$ with constant probability, where $\text{OPT}_{\text{Aug}}$ is the optimum augmentation cost. 
\end{theorem}

Theorem \ref{thm:pqFGC} follows from Theorem \ref{thm:pqFGCAug}: we obtain an initial min-cost $p$-edge-connected subgraph with a $2$-factor loss and iteratively augment to achieve $(p, i)$-Flex-Connected in the $i$'th iteration for each $i=1,\ldots, q$ using the algorithm in Theorem \ref{thm:pqFGCAug}. See Section \ref{sec:fgc-thm-proof} for details. We focus on Theorem \ref{thm:pqFGCAug} henceforth. 

Ibrahimpur and V\'{e}gh \cite{IV26-j} gave a randomized $O(\log{n})$-approximation for \pqFGC. It also implies a randomized $O(\log{n})$-approximation for \pqFGCAug by giving a cost of zero to the edges of $H$. Thus, if $qL(p, q)>\log{n}$, then Theorem \ref{thm:pqFGCAug} is implied by Ibrahimpur-V\'{e}gh's result. 
Consequently, we may henceforth assume that the input instance is feasible and $qL(p, q)\le \log{n}$. 
The rest of the proof of Theorem \ref{thm:pqFGCAug} is via the following core result that is a deterministic algorithm. We note that randomization in the algorithm of Theorem \ref{thm:pqFGCAug} arises only from the use of Ibrahimpur-V\'{e}gh's algorithm for the case of $qL(p, q)>\log{n}$.

\begin{theorem}\label{thm:pqFGCAug-with-runtime}
    There exists a deterministic algorithm for \pqFGCAug that returns a feasible augmentation $F$ with $c(F)\le 2L(p, q)OPT_{aug}$, where $OPT_{aug}$ is the cost of an optimum solution and its run-time is $qL(p, q)$ times a polynomial in the input size. 
\end{theorem}

\begin{remark}
The above results can be strengthened along three fronts: 
Firstly, randomization is not necessary. Both Theorems \ref{thm:pqFGCAug} and \ref{thm:pqFGC} can be strengthened to be a ``deterministic'' algorithm with the same approximation guarantees. For this, we need an alternative algorithm when $L(p, q)$ is large instead of using Ibrahimpur-V\'{e}gh's result. 
Secondly, the approximation factor of the deterministic algorithm can be tightened to be $2L(p,q)$ for the purposes of Theorem \ref{thm:pqFGCAug} (and not just $O(qL(p, q))$). This leads to an approximation guarantee of $2\sum_{j=0}^q L(p, j)$ in Theorem \ref{thm:pqFGC}. 
Thirdly, the approximation factor of Theorem \ref{thm:pqFGC} can be shown to be relative to 
a relaxation of Chekuri-Jain's LP-relaxation for \pqFGC \cite{CJ25}. I.e., we can drop certain constraints from Chekuri-Jain's LP-relaxation and still achieve these approximation guarantees. 
We skip these strengthenings in the interests of brevity and concise presentation. 
The rest of the paper is devoted to proving Theorem \ref{thm:pqFGCAug-with-runtime}.    
\end{remark}

\subsection{Techniques}\label{sec:techniques}
Our approach to prove Theorem \ref{thm:pqFGCAug-with-runtime} is inspired by the works of \cite{BCHI24} and \cite{CJ25}. We briefly describe their approaches. Boyd, Cheriyan, Haddadan, and Ibrahimpur \cite{BCHI24} obtained an $O(q\log{n})$-approximation for \pqFGC by a two-phase algorithm, where $n:=|V|$ is the number of vertices in the input graph: the first phase finds an initial subset of edges and the second phase augments the current solution in at most $q$ iterations to arrive at a feasible solution; each iteration corresponds to solving a set-cover instance with $n^{O(1)}$-elements, thereby incurring an $O(\log{n})$-approximation per iteration. Our result in Theorem \ref{thm:pqFGCAug-with-runtime} can be viewed as improving the $O(\log{n})$-approximation per phase to $2L(p, q)$-approximation per phase. 
Chekuri and Jain \cite{CJ25} adapted the augmentation approach to obtain an approximation that depends only on $p$ and $q$ for $(p=2, q)$-FGC, $(p, q=2)$-FGC, and $(p, q=3)$-FGC, i.e., for certain regimes of $p$ and $q$. 
Recall that the goal is to cover deficient-cuts, i.e., cuts in $\cT$. They view this problem itself as a sequence of augmentation problems: 
in iteration $i$, they augment to cover all uncovered cuts $B\in \cT$ with $|\delta_S(B)|\le i$. They show that 
the set of cuts to be covered in iteration $i$ has an \emph{uncrossable structure} for certain values of $(p, q)$ and exploit this to obtain the stated approximation. However, the uncrossability structure breaks down outside certain regimes of $(p, q)$. Extensions of the primal-dual method to pliable families have overcome the need for uncrossability in several settings \cite{Ban25, BCGI23, Nut25, Nut25-tight}, but pliability also has its limitations---see \cite{BCGI23}. Instead of crossing/pliable structure on the family of cuts to be covered, our approach shows a different structure. 

Our high-level plan to prove Theorem \ref{thm:pqFGCAug-with-runtime} is as follows: The number of deficient-cuts---i.e., $|\cT|$---need not be bounded by a function of $p$ and $q$. Nevertheless, we show a representation of all such cuts via an $L(p, q)$-sized collection of capacitated digraphs. This is the key structural result, so we set up the background and state the result. 
A \emph{capacitated digraph} $D=(V, A, w)$ is given by a vertex set $V$, arc set $A$, and non-negative integer capacities $w: A\rightarrow \Z_{\ge 0}$. Fix a root $r\in V$. For an arc set $F$ and a subset $B\subseteq V$, we define $\delta^{in}_F(B):=\{uv\in F: v\in B, u\in V\setminus B\}$, $d^{in}_D(B):=\sum_{a\in \delta^{in}_A(B)}w_a$, and the min rooted-cut capacity $\lambda_D:=\min\{d^{in}_D(B): \emptyset\neq B\subseteq V\setminus\{r\}\}$. A non-empty subset $B\subseteq V\setminus \{r\}$ is said to be a min rooted-cut in $D$ if $d^{in}_D(B)=\lambda_D$. 
We show the following structural result. 
\begin{theorem}\label{thm:representatives}
    There is an algorithm that takes as input 
    non-negative integers $p$ and $q$ and 
    a graph $H=(V, E=S\uplus U)$ where edges in $S$ are safe and edges in $U$ are unsafe such that $H$ is \pqminusoneFlexC
    and returns a collection $\cC$ of capacitated digraphs such that 
    \begin{enumerate}
        \item $|\cC|\le L(p, q)$ and 
        \item $B\in \cT$ if and only if $B$ is a minimum rooted-cut of some $D\in \cC$. 
    \end{enumerate}
    Moreover, the running time of the algorithm is $qL(p, q)$ times a polynomial in the input size. 
\end{theorem}

Theorem \ref{thm:representatives} shows that there exists an $L(p,q)$-sized collection of representatives for deficient-cuts. Once we have such a collection, the rest of the algorithm to prove Theorem \ref{thm:pqFGCAug-with-runtime} is straightforward, and we describe it now. Consider a capacitated digraph $D\in \cC$. The relevant problem that we need to solve now for each $D$ is the following: 
\begin{align}
\min\left\{c(F): F\subseteq N, |\delta_F(B)|\ge 1\ \forall \text{ min rooted-cut }B \text{ of } D\right\}.\label{eq:rooted-cut-covers}
\end{align}
We obtain a $2$-approximation for this problem as follows: let $A'$ denote the set of arcs obtained by bidirecting all edges of $N$ and setting costs $c'(e_1)=c'(e_2)=c(e)$ for both orientations $e_1$ and $e_2$ of $e$. Consider the following problem: 
\begin{align}
\min\left\{c'(F'): F'\subseteq A', |\delta^{in}_{F'}(B)|\ge 1\ \forall \text{ min rooted-cut }B \text{ of } D\right\}. \label{eq:rooted-cut-augmentation}
\end{align}
This is the min rooted-connectivity augmentation problem (i.e., we would like to find a min-cost set of arcs to add to increase the min rooted-cut capacity by at least one) and it is solvable in polynomial time via the results of Edmonds. An optimal solution $F_{opt}$ for the input instance of \pqFGCAug yields a feasible solution $F'_{opt}$ for min rooted-connectivity augmentation problem associated with $D$ with $c'(F'_{opt})=2c(F_{opt})=2OPT_{aug}$: take $F'_{opt}$ to be both orientations of every edge in $F_{opt}$. Consequently, solving \eqref{eq:rooted-cut-augmentation} gives an optimum solution $F_D'$ with $c'(F_D')\le c'(F'_{opt})\le 2OPT_{aug}$. Now, consider the set $F_D\subseteq N$ consisting of those edges with at least one orientation in $F_D'$. Then, $F_D$ is feasible for \eqref{eq:rooted-cut-covers} and has $c(F_D)\le c'(F_D')\le 2OPT_{aug}$. Thus, the set $F:=\cup_{D\in \cC}F_D$ is feasible for \pqFGCAug and has cost $c(F)\le 2L(p, q) OPT_{aug}$; the last inequality is because $|\cC|\le L(p, q)$. 

To prove the structural result stated in Theorem \ref{thm:representatives}, we start from the capacitated digraph $D_0$ obtained by bidirecting all edges of $H$ and declare all arc capacities to be unit. We recursively increase capacities of certain safe-arcs to infinity to obtain the desired collection. To determine which safe-arc capacities to increase, we rely on a maximum packing of arborescences in the current capacitated digraph. 

The proof of Theorem \ref{thm:representatives} is inspired by Chekuri and Jain's techniques for \stpqFGC. In \stpqFGC, the input consists of the input of \pqFGC and additionally consists of a distinct pair of vertices $s$ and $t$ and the goal is to find a minimum cost subset $F$ of edges such that $|\delta_{S\cap F}(B)|\ge p$ or $|\delta_F(B)|\ge p+q$ for every $s\in B\subseteq V-\{t\}$. 
Chekuri and Jain \cite{CJ25} designed an algorithm for \stpqFGC that runs in time $n^{O(p+q)}$ and achieves an approximation factor of $(p+q)^{O(p)}$ for every $p$ and $q$ satisfying $p+q>pq/2$. They start with an initial solution that is feasible for $(p=0, q)$, and in the $i$'th iteration, augment the current solution to obtain a feasible solution for $(p=i, q)$. To solve the augmentation problem, they show a representation family for deficient-cuts by packing $s$-$t$ paths and considering $(i-1)$-sized subsets of these packed paths. In our setting, we will pack arborescences and consider $(p-1)$-sized subsets of packed arborescences. Our construction of representatives is via recursion in contrast to Chekuri-Jain: for each $(p-1)$-sized subset, we increase the capacities of all safe arcs in the union of arborescences that are \emph{not} indexed by that subset; after increasing, we choose to store or discard or recurse from that digraph based on its min rooted-cut capacity relative to the current digraph. The recursive increase of safe-arc capacities is the key technical departure in our approach compared to the techniques of Chekuri-Jain.

\section{Digraph Preliminaries}\label{sec:prelims}
We recall the relevant preliminaries on digraphs that we need. Throughout, graph denotes an undirected graph. Bidirecting an undirected edge $uv$ replaces it with opposite arcs $u\rightarrow v$ and $v\rightarrow u$. We denote $[k]:=\{1, 2, \ldots, k\}$. For a vector $c: N\rightarrow \R_{\ge 0}$ and a subset $F\subseteq N$, we write $c(F):=\sum_{e\in F}c_e$. 

We recap the digraph notions defined in Section \ref{sec:techniques}
Let $D=(V, A, w: A\rightarrow \Z_{\ge 0})$ be a  \emph{capacitated digraph} given by vertex set $V$, arc set $A$, and non-negative integer capacities $w: A\rightarrow \Z_{\ge 0}$. Fix a root $r\in V$. For an arc set $F$ and a vertex-subset $B\subseteq V$, we define $\delta^{in}_F(B):=\{uv\in F: v\in B, u\in V\setminus B\}$, $d^{in}_D(B):=\sum_{a\in \delta^{in}_A(B)}w_a$, and the min rooted-cut capacity $\lambda_D:=\min\{d^{in}_D(B): \emptyset\neq B\subseteq V\setminus\{r\}\}$. A non-empty subset $B\subseteq V\setminus \{r\}$ is said to be a min rooted-cut if $d^{in}_D(B)=\lambda_D$. Min rooted-cut of a given capacitated digraph can be computed in polynomial time \cite{schrijver-comb-opt-book}. 

Let $D=(V, A, w: A\rightarrow \Z_{\ge 0})$ be a  \emph{capacitated digraph}. 
A subset $T\subseteq A$ is an \emph{arborescence} if the undirected version of the digraph $(V, T)$ is a spanning tree and for every vertex $u\in V$, there exists a unique path from $r$ to $u$ in the digraph $(V, T)$. 
A capacitated digraph $D=(V, A, w:A\rightarrow \Z_{\ge 0})$ \emph{packs} $k$ arborescences if there exist arborescences $T_1, T_2, \ldots, T_k\subseteq A$ such that each arc $a$ is present in at most $w_a$ arborescences. Such a collection of arborescences is called a packing of $k$ arborescences in $D$. 
We need the following result due to Edmonds. 

\begin{theorem}\cite{schrijver-comb-opt-book}\label{thm:arb-packing}
The maximum number of arborescences that can be packed in a capacitated digraph $D$ is equal to $\lambda_D$. 
    Given a digraph $D$ in which all capacities are either unit or infinity with finite $\lambda_D$, there exists a polynomial-time algorithm to compute a packing of $\lambda_D$ arborescences in $D$. 
\end{theorem}

We briefly address the second statement in the above theorem, which is derived from the fact that given a digraph $D$ with unit-capacities, there exists a polynomial-time algorithm to compute a packing of $\lambda_D$ arborescences in $D$ (see \cite{schrijver-comb-opt-book}). To derive the second statement, we proceed as follows: 
If all capacities are unit or infinity, but $\lambda_D$ is finite, then we may replace each  infinite-capacity arc with $1+\lambda_D$ parallel unit-capacity arcs. We note that $\lambda_D$ is at most the number of arcs in $D$ and hence, this transformation gives a unit-capacity graph $D'$ whose size is at most polynomial in the number of arcs in $D$. Moreover, $\lambda_{D'}=\lambda_D$ and a packing of $\lambda_{D'}$ arborescences in $D'$ is also a packing of $\lambda_D$ arborescences in $D$. Thus, we infer the result using the polynomial-time algorithm for computing arborescence packing in unit-capacity digraphs. 

We need the result that the minimum cost rooted-connectivity augmentation problem can be solved in polynomial time. We define the problem and state the result. 

\begin{center}
\fbox{%
\begin{minipage}{0.95\textwidth}
\textbf{\rootedCAug:}

\textbf{Given:} Capacitated digraph $D=(V, A, w: A\rightarrow \Z_{\ge 0})$, root $r\in V$, set $A'$ of arcs on vertex set $V$ with non-negative costs $c': A'\rightarrow \R_{\ge 0}$ and unit-capacities. 

\textbf{Goal:} Find a subset $F'\subseteq A'$ with minimum $\sum_{a\in F'}c'_a$ such that $\lambda_{D+F'}\ge \lambda_D+1$. 

\end{minipage}}
\end{center}

\begin{theorem}\label{thm:rootedCAug-is-poly-time}
    There is a deterministic polynomial-time algorithm to solve \rootedCAug. 
\end{theorem}
\begin{proof}
Feasibility verification can be done in polynomial time since min rooted-cut can be computed in polynomial time. By Theorem \ref{thm:arb-packing}, $\lambda_{D+F'}\ge \lambda_D+1$ if and only if $D+F'$ can pack $\lambda_D+1$ arborescences. Consider the digraph $D''=(V, A\cup A', w'': A\cup A'\rightarrow \Z_{\ge 0})$ where the capacity function $w''$ is given by $w''(a)=w(a)$ if $a\in A$ and $w''(a)=1$ if $a\in A'$. Assign cost function $c'': A\cup A'\rightarrow \R_{\ge 0}$ by setting $c''(a)=0$ if $a\in A$ and $c''(a)=c'(a)$ if $a\in A'$. Then, the problem is equivalent to finding a min-cost subgraph of $D''$ that can pack $k$ arborescences. The latter problem is indeed solvable in polynomial time (see Theorem 53.6 in Schrijver's book \cite{schrijver-comb-opt-book}). 
\end{proof}

\section{Representing deficient-cuts}\label{sec:reps}
We prove Theorem \ref{thm:representatives} in this section. 
We recall that the input is a graph $H=(V, E=S\uplus U)$ with $S$ being the safe edges and $U$ being the unsafe edges, and non-negative integers $p$ and $q$ such that $H$ is \pqminusoneFlexC. 
For convenience, we use $k=p+q-1$ throughout this section. If $H$ is \pqFlexC, then $\cT=\emptyset$ and we can simply return $\cC=\emptyset$. For the rest of the section, we assume that $H$ is not \pqFlexC but is only \pqminusoneFlexC. 

We will work with capacitated digraphs $D=(V, A, w: A\rightarrow \Z_{\ge 0})$ in which each arc is classified as either safe or unsafe. Our algorithm will recursively increase the capacity of certain arcs to infinity and we term this \emph{promoting}: promoting an arc means that its capacity is being increased to $\infty$. We emphasize a subtle but important difference: in contrast to undirected graphs where inflating the capacity of an edge to infinity amounts to contracting that edge for the purposes of min-cut/connectivity, in directed graphs, setting an arc capacity to infinity does not amount to arc contraction for the purposes of min rooted-cut. 

Let $D_0$ be the capacitated digraph obtained from the undirected graph $H=(V, E=S\uplus U)$ by bidirecting all of its edges and setting the capacity of all arcs to be unit. We carry over the safe/unsafe classification from $H$ to $D_0$: both orientations of each safe edge are safe and both orientations of each unsafe edge are unsafe. Recall that we have a fixed root vertex $r\in V$. Our algorithm is Construct-Representatives$(D_0)$ given below. 
Let $\cC_0$ denote the collection returned by Construct-Representatives$(D_0)$. 
All digraphs constructed by this recursive algorithm have the same vertex set, root, arcs, and classification of safe and unsafe arcs as that of $D_0$, but they may differ from $D_0$ only in arc capacities. 

\medskip
\noindent\fbox{
\begin{minipage}{0.95\linewidth}
\textbf{Construct-Representatives$(D)$}
\begin{enumerate}\setlength{\itemsep}{3pt}
 \item Initialize $\cC_D\gets\emptyset$.
 \item Find a packing $T_1,\ldots,T_{\lambda_D}$ of arborescences in $D$ using Theorem~\ref{thm:arb-packing}.
 \item For each $I\subseteq[\lambda_D]$ with $|I|=p-1$:
 \renewcommand{\labelenumii}{(\roman{enumii})}
\begin{enumerate}\setlength{\itemsep}{3pt}
  \item Create $D_I$ from $D$ by promoting every safe arc in
  $\bigcup_{i\in[\lambda_D]\setminus I}T_i$.
  \item Compute $\lambda_{D_I}$. 
  \item If $\lambda_{D_I}=\lambda_D$, add $D_I$ to $\cC_D$.
  \item If $\lambda_D<\lambda_{D_I}\le k$, add the collection returned by Construct-Representatives$(D_I)$ to $\cC_D$.
  \item If $\lambda_{D_I}>k$, discard $D_I$.
 \end{enumerate}
 \item Return $\cC_D$.
\end{enumerate}
\end{minipage}
}
\medskip

To analyze the algorithm, it is convenient to associate a recursion tree with the execution of the algorithm: The root node of the recursion tree is $D_0$. Suppose the algorithm is called on input $D$. For each $D_I$ created during the execution of Step 3 of the algorithm, we create a node labeled by $D_I$ as a child of $D$ and say that $D$ is the parent of $D_I$. A node with children is an \emph{internal-node} and every other node is a \emph{leaf}. 
A leaf node $D_I$ discarded in Step 3(v) is a \emph{discarded-leaf} and 
every other leaf is a \emph{representative-leaf}. We note that $\cC_0$ is exactly the set of representative-leaves. 

We show that Construct-Representatives$(D_0)$ terminates, and bound the number of representative-leaves and its run-time in Section \ref{sec:termination-and-number}. Next, we show the correspondence between deficient-cut and min rooted-cuts in digraphs in $\cC_0$ in Section \ref{sec:def-and-min-rooted-cuts}.
Lemma \ref{lem:size-bound} in Section \ref{sec:termination-and-number} and Lemmas \ref{lem:survivors-are-deficient} and \ref{lem:deficient-cuts-survive} in Section \ref{sec:def-and-min-rooted-cuts} together complete the proof of Theorem \ref{thm:representatives}.

\subsection{Termination and number of returned digraphs}\label{sec:termination-and-number}
In this section, we show that the algorithm terminates and bound $|\cC_0|\le L(p, q)$. 

\begin{proposition}\label{prop:easy-obs}
    Let $D$ be a capacitated digraph that occurs in the recursion tree. The following hold: 
    \begin{enumerate}[label={(\arabic*)}]
        \item Every arc in $D$ has capacity either $1$ or $\infty$. 
        \item $\lambda_D\ge p$. 
        \item If $D$ is an internal-node or a representative-leaf, then $\lambda_D\le k$. 
        \item If $D'$ is a child of $D$ and $D'$ is an internal-node, then $\lambda_{D'}\ge\lambda_D + 1$. 
    \end{enumerate}
\end{proposition}
\begin{proof}
\begin{enumerate}[label={(\arabic*)}]
    \item Since every $D$ that occurs in the recursion tree is obtained by promoting some arcs of $D_0$ and all arcs of $D_0$ have capacity one, we have that every arc in $D$ has capacity either $1$ or $\infty$. 
    \item Since $H=(V, E)$ is \pqminusoneFlexC, every non-empty proper subset $B\subsetneq V$ has at least $p$ edges. Consequently, $d^{in}_{D_0}(B)=|\delta_E(B)|\ge p$ for every $B\subseteq V\setminus\{r\}$. Thus, $\lambda_{D_0}\ge p$. Since every $D$ that occurs in the recursion tree is obtained by promoting some arcs of $D_0$, we have that $\lambda_D\ge \lambda_{D_0}\ge p$. 
    \item This is by construction of internal-nodes and representative-leaves. 
    \item Suppose $\lambda_{D'}=\lambda_D$. Then, by Step 3(iii) of the algorithm, $D'$ has to be a representative-leaf and cannot be an internal-node, a contradiction. 
\end{enumerate}
\end{proof}

The next lemma implies that the recursive algorithm terminates. 
\begin{lemma}\label{lem:termination}
Every root to leaf path in the recursion tree has at most $q$ internal-nodes. 
\end{lemma}
\begin{proof}
    Let $D_1$ be a leaf in the recursion tree with $D'$ being its parent. Suppose the root to leaf path has $t$ internal-nodes. Then, by Proposition \ref{prop:easy-obs}(4), $\lambda_{D'}\ge \lambda_{D_0}+t$. By Proposition \ref{prop:easy-obs}(3), $\lambda_{D'}\le k$. Thus, $t\le k-\lambda_{D_0}$. By Proposition \ref{prop:easy-obs}(3), we have that $\lambda_{D_0}\ge p$ and hence, $t\le k-p=q-1$. 
\end{proof}

Lemma \ref{lem:termination} implies that the recursive algorithm terminates since the branching factor at each internal-node is finite. We next bound the runtime and size of the family $\cC_0$ returned by the algorithm. 

\begin{lemma}\label{lem:size-bound}
The algorithm can be implemented to run in time $qL(p, q)$ times polynomial in the input size and moreover, 
\[
|\cC_0|\le L(p, q)
\]
\end{lemma}
\begin{proof}
    The computation at each internal-node $D$ is polynomial time since $\lambda_D$ and a packing of arborescences in $D$ can be computed in polynomial time (by Proposition \ref{prop:easy-obs}(1) and Theorem \ref{thm:arb-packing}). Lemma \ref{lem:termination} implies that the recursive algorithm terminates since the branching factor at each internal-node is finite. 
    In order to bound the run-time and the size of $\cC_0$, it suffices to bound the number of leaves in the recursion tree. 

    We show by induction on $k-\lambda_D$ that the number of leaves in the sub-tree rooted at $D$ is at most $\prod_{\ell=\lambda_D}^{k}\binom{\ell}{p-1}$. The base case is trivial. We show the induction step. Firstly, $D$ has $\binom{\lambda_D}{p-1}$ children. Each child $D_I$ of $D$ is either a leaf or an internal-node. If $D_I$ is an internal-node, then it has $\lambda_{D_I}\ge\lambda_D+1$ by Proposition \ref{prop:easy-obs}(4). Hence, by induction, the sub-tree rooted at $D_I$ has at most $\prod_{\ell=\lambda_{D_I}}^{k}\binom{\ell}{p-1}\le \prod_{\ell=\lambda_D+1}^{k}\binom{\ell}{p-1}$ leaves. Thus, the total number of leaves in the subtree rooted at $D$ is at most $\prod_{\ell=\lambda_D}^{k}\binom{\ell}{p-1}$. 

    The above statement implies that $|\cC_0|$ is at most the number of leaves in the sub-tree rooted at $D_0$ which is at most  $\prod_{\ell=\lambda_{D_0}}^{k}\binom{\ell}{p-1}\le L(p, q)$ since $\lambda_{D_0}\ge p$ by Proposition \ref{prop:easy-obs}(2). 

    Lemma \ref{lem:termination} implies that the depth of the recursion tree is at most $q$ and consequently, the total number of nodes in the recursion tree is at most $qL(p, q)$. This implies the stated run-time bound. 
\end{proof}

\subsection{Deficient-cuts and min rooted-cuts}\label{sec:def-and-min-rooted-cuts}
In this section, we show that a set $B\subseteq V\setminus\{r\}$ is a deficient-cut if and only if $B$ is a min rooted-cut of some capacitated digraph $D\in \cC_0$. 
We first show that every min rooted-cut of every representative-leaf is a deficient-cut. 
\begin{lemma}\label{lem:survivors-are-deficient}
    If $D\in \cC_0$, then $\lambda_{D}=k$ and moreover, every min rooted-cut of $D$ is in $\cT$. 
\end{lemma}
\begin{proof}
    Since $D\in \cC_0$, it has a parent $D'$ in the recursion tree with $\lambda_{D}=\lambda_{D'}\le k$. Let $B$ be a min rooted-cut of $D$. Recall that $H=(V, E=S\uplus U)$ where edges in $S$ are safe and edges in $U$ are unsafe and we need to show that $B\in \cT$, i.e., $|\delta_{S}(B)|\le p-1$ and $|\delta_{E}(B)|=p+q-1$. 
    
    Since $B$ is a min rooted-cut in $D$, we have that $\lambda_{D}=d^{in}_D(B)\le k$. Thus, there are no promoted arcs entering $B$ in $D$ as well as $D'$. 

    We first show that $|\delta_{S}(B)|\le p-1$. 
    Since $D$ is a child of $D'$, there exists a packing $T_1, T_2, \ldots, T_{\lambda_{D'}}$ of arborescences in $D'$, a subset $I\subseteq[\lambda_{D'}]$ with $|I|=p-1$ such that $D=D_I$. 
    We need the following claim. 

    \begin{claim}\label{claim:arborescence}
    We have the following: 
    \begin{enumerate}[label={(\arabic*)}]
        \item $|\delta^{in}_{T_i}(B)|=1\ \forall\ i\in [\lambda_{D'}]$, 
        \item $d^{in}_{D'}(B)=\lambda_{D'}$, 
        \item every arc entering $B$ in $D'$ is present in exactly one of the arborescences $T_1, T_2, \ldots, T_{\lambda_{D'}}$. 
    \end{enumerate}
        
    \end{claim}
    \begin{proof}
    We have that $|\delta^{in}_{T_i}(B)|\ge 1$ for every $i\in [\lambda_{D'}]$ by definition of arborescences. Hence, 
    \begin{align*}
        \lambda_{D'}
        \le \sum_{i=1}^{\lambda_{D'}}|\delta_{T_i}^{in}(B)|
        \le d^{in}_{D'}(B)
        =d^{in}_D(B)
        =\lambda_D
        =\lambda_{D'}, 
    \end{align*}
    where the second inequality is because all arcs in $\delta^{in}_{D'}(B)$ have capacity one since no promoted arcs enter $B$ in $D'$ and consequently, the family $\left\{\delta^{in}_{T_i}(B)\right\}_{i=1}^{\lambda_{D'}}$ is pairwise disjoint. The first equality is because there are no promoted arcs entering $B$ in $D'$ and $D$. Thus, the above sequence of inequalities is tight. Moreover, 
    $|\delta_{T_i}^{in}(B)|=1$ for each $i\in [\lambda_{D'}]$ and  
    $d^{in}_{D'}(B)=\lambda_{D'}$ and thus, $B$  is a min rooted-cut in $D'$ as well. 
    Finally, every arc entering $B$ in $D'$ has to be present in exactly one of the arborescences $T_1, T_2, \ldots, T_{\lambda_{D'}}$: if there exists an arc not present in any of them, then $d^{in}_{D'}(B)>\sum_{i=1}^{\lambda_{D'}}|\delta_{T_i}^{in}(B)|$, a contradiction; moreover, an arc cannot be present in multiple arborescences, since every arc entering $B$ in $D'$ is not promoted and hence, has capacity one. 
    \end{proof}

    We recall that $|\delta_{S}(B)|$ is at most the number of safe arcs entering $B$ in $D'$. 
    By Claim \ref{claim:arborescence}(3), every arc entering $B$ in $D'$ is present in exactly one of the arborescences $T_1, T_2, \ldots, T_{\lambda_{D'}}$. 
    Since no promoted arcs enter $B$ in $D'$, it follows that there are no safe arcs entering $B$ in $\cup_{i\in [\lambda_D]\setminus I}T_i$. Consequently, all safe arcs entering $B$ in $D'$ are in $\cup_{i\in I}T_i$. 
    Claim \ref{claim:arborescence}(1) implies that the number of safe arcs entering $B$ in $\cup_{i\in I}T_i$ is at most $|I|=p-1$. 
    Hence, $|\delta_{S}(B)|\le p-1$. 

    Next, we show that $|\delta_E(B)|=p+q-1=k$. We have already shown that $|\delta_{S}(B)\le p-1$. Since $H=(V, E=S\uplus U)$ is \pqminusoneFlexC, we have that $k\le |\delta_{E}(B)|=d^{in}_{D_0}(B)\le d^{in}_{D}(B)=\lambda_{D}\le k$ and hence, $|\delta_{E}(B)|=k$. This sequence of inequalities also implies that $\lambda_D=k$. 

\end{proof}

Next, we show that if $B$ is a deficient-cut, then it is a min-rooted cut in some representative leaf. We need an intermediate claim. Let $B\in \cT$. 
We say that $B$ \emph{survives} in a capacitated digraph $D$ if $\delta_D^{in}(B)$ has no promoted arcs. 
\begin{claim}\label{claim:survival}
    Let $B\in \cT$. 
    Suppose $B$ survives in an internal-node $D$. Then, $B$ survives in at least one of the children $D_I$ of $D$ and moreover, $D_I$ is not discarded. 
\end{claim}
\begin{proof}
    Since there are no promoted arcs entering $B$ in $D$, the number of safe arcs entering $B$ in $D$ is equal to the number of safe arcs entering $B$ in $D_0$. The number of safe arcs  entering $B$ in $D_0$ is exactly $|\delta_{S}(B)|\le p-1$ since $B\in \cT$. Hence, the number of safe arcs entering $B$ in $D$ is at most $p-1$. 

    Since $D$ is an internal node, the children of $D$ are obtained by 
    considering a packing $T_1, T_2, \ldots, T_{\lambda_{D}}$ of arborescences in $D$; for each subset $I\subseteq[\lambda_{D}]$ with $|I|=p-1$, we obtain a child $D_I$ from $D$ by promoting the safe arcs in $\cup_{i\in [\lambda_D]\setminus I}T_i$. Since the number of safe arcs of $D$ entering $B$ is at most $p-1$, it follows that there exists an $I\subseteq[\lambda_D]$ with $|I|=p-1$ such that all safe arcs entering $B$ in $D$ are not promoted in $D_I$. Thus, $B$ survives in $D_I$. 

    Next, we show that $D_I$ is not discarded. 
    Since $B\in \cT$, we have that $k=|\delta_E(B)|=d^{in}_{D_0}(B)$. Moreover, $d^{in}_{D_0}(B)=d^{in}_{D_I}(B)$ since $B$ survives in $D_I$ and hence, has no promoted arcs entering it in $D_I$. Thus, $\lambda_{D_I}\le d^{in}_{D_I}(B)=k$. Therefore, $D_I$ is not discarded. 
\end{proof}

We now use Claim \ref{claim:survival} and Lemma \ref{lem:survivors-are-deficient} to show that every deficient-cut is a min-rooted cut in some representative-leaf. 

\begin{lemma}\label{lem:deficient-cuts-survive}
    If $B\in \cT$, then $B$ is a min rooted-cut of $D$ for some $D\in \cC_0$. 
\end{lemma}
\begin{proof}
    Initially, $B$ survives in $D_0$. By applying Claim \ref{claim:survival} inductively, $B$ survives in some representative-leaf $D$. We recall that every representative-leaf is in $\cC_0$. Thus, $B$ survives in a representative leaf $D\in \cC_0$. It suffices to show that $B$ is a min rooted-cut of $D$. We have that 
    $d^{in}_D(B)=d^{in}_{D_0}(B)$ since $B$ survives in $D$ and hence, has no promoted arcs entering it in $D$. Moreover, $d^{in}_{D_0}(B)=|\delta_E(B)|=k$ since $B\in \cT$. 
    We know that $\lambda_D=k$ by Lemma \ref{lem:survivors-are-deficient}. Thus, 
    \[
    k=\lambda_D \le d^{in}_D(B) = d^{in}_{D_0}(B)=|\delta_E(B)|=k
    \]
    implies that $\lambda_D=d^{in}_D(B)$ and hence, $B$ is a min rooted-cut of $D$. 
\end{proof}

\section{Proof of Theorem \ref{thm:pqFGCAug-with-runtime}}\label{sec:covering-reps}

\begin{proof}[Proof of Theorem \ref{thm:pqFGCAug-with-runtime}]
We use the algorithm below. 

\medskip
\noindent\fbox{
\begin{minipage}{0.95\linewidth}
\textbf{Augment$(H=(V, E=S\uplus U), p, q, N, c: N\rightarrow \R_{\ge 0})$}
\begin{enumerate}\setlength{\itemsep}{3pt}
 \item Construct the capacitated digraph $(V, A')$ by bidirecting every edge of $N$. Set the capacity of each $a\in A'$ to be unit and define the cost function $c': A'\rightarrow \R_{\ge 0}$ by setting $c'(a)=c(e)$ for both orientations $a$ of $e$ for all $e\in N$. 
 \item Apply Theorem \ref{thm:representatives} to input $(H, p, q)$ to obtain the collection $\cC$. 
 \item For each $D\in \cC$: 
 \renewcommand{\labelenumii}{(\roman{enumii})}
\begin{enumerate}\setlength{\itemsep}{3pt}
  \item Compute $F_D'$ as an optimal solution to the \rootedCAug instance $(D, A', c')$ using Theorem \ref{thm:rootedCAug-is-poly-time}. 
  \item Set $F_D$ to be the set of edges $e\in N$ such that at least one orientation of $e$ is in $F_D'$.
\end{enumerate}
 \item Return $F=\cup_{D\in \cC}F_D$.
\end{enumerate}
\end{minipage}
}
\medskip

We first bound the run-time of the algorithm. The algorithm of Theorem \ref{thm:representatives} constructs a collection $\cC$ of at most $L(p,q)$ capacitated digraphs in time $qL(p, q)$ times a polynomial in the input size. For each digraph $D\in \cC$, Step 3(i) constructs a \rootedCAug instance whose size is polynomial in the input size. The algorithm of Theorem \ref{thm:rootedCAug-is-poly-time} runs in polynomial time. Step 3(ii) is also polynomial time. Hence, the total run-time is $qL(p, q)$ times polynomial in the input size. 

Next, we show that the returned solution $F$ is feasible. Let $B\in \cT$ be a deficient cut. We recall that the input instance of \pqFGCAug is feasible. Hence, $|\delta_N(B)|\ge 1$. By Theorem \ref{thm:representatives}, there exists $D\in \cC$ such that $B$ is a min rooted-cut in $D$. Consider the solution $F_D'$ to the \rootedCAug instance given by $(D, A', c')$. Since $\lambda_{D+F_D'}\ge \lambda_D+1$ and $B$ is a min rooted-cut in $D$, we have that $|\delta_{F_D'}^{in}(B)|\ge 1$. Consequently, $|\delta_{F_D}(B)|\ge 1$. 

Next, we bound the approximation factor. We note that $c(F_D)\le c'(F_D')$ for each $D\in \cC$. 
Let $F_{opt}$ be an optimum solution for the input instance of \pqFGCAug. We will show that $c'(F_D')\le 2 c(F_{opt})$ for each $D\in \cC$. Let $F_{opt}'$ be obtained by bidirecting every edge in $F_{opt}$. We observe that $c'(F_{opt}')=2c(F_{opt})$. 

We claim that for each $D\in \cC$, $F_{opt}'$ is a feasible solution to the \rootedCAug instance given by $(D, A', c')$. Let $D\in \cC$. Consider a subset $B\subseteq V\setminus \{r\}$. We need to show that $d^{in}_{D+F_{opt}'}(B)\ge \lambda_D + 1$. For the sake of contradiction, suppose $d^{in}_{D+F_{opt}'}(B)\le \lambda_D$. We know that $\lambda_D\le d^{in}_{D}(B)\le d^{in}_{D+F_{opt}'}(B)\le \lambda_D$. Hence, $B$ is a min rooted-cut of $D$. By Theorem \ref{thm:representatives}, we have that $B\in \cT$. Hence, $|\delta_{F_{opt}}(B)|\ge 1$. Therefore, $|\delta^{in}_{F_{opt}'}(B)|=|\delta_{F_{opt}}(B)|\ge 1$. Consequently, $d^{in}_{D+F_{opt}'}(B)=d^{in}_D(B)+ |\delta^{in}_{F_{opt}'}(B)|\ge \lambda_D+1$, a contradiction. 

We have shown that $F_{opt}'$ is a feasible solution to the \rootedCAug instance given by $(D, A', c')$. We recall that $F_D'$ is an optimum to the \rootedCAug instance given by $(D, A', c')$. 
Hence, $c(F_D)\le c'(F_D')\le c'(F_{opt}') = 2c(F_{opt})=2OPT_{aug}$. Finally, we have that $c(F)\le \sum_{D\in \cC}c(F_D)\le 2|\cC|OPT_{aug}\le 2L(p, q)OPT_{aug}$, where the last inequality is because $|\cC|\le L(p, q)$. 
    
\end{proof}
\section{Proof of Theorem \ref{thm:pqFGC}}\label{sec:fgc-thm-proof}
\begin{proof}[Proof of Theorem \ref{thm:pqFGC}]

If $\alpha(p, q)>\log{n}$, then we use Ibrahimpur-V\'{e}gh's $O(\log{n})$-approximation. Henceforth, we assume that the input instance has $\alpha(p, q)\le \log{n}$. 

Let $(G=(V, E=S\uplus U), p, q, c: E\rightarrow \R_{\ge 0})$ be the input instance. In the $0$'th phase, we obtain a $2$-approximate minimum-cost $p$-edge-connected spanning subgraph of $G$ by ignoring the safe/unsafe classification, i.e., we obtain a two-approximation for the following problem:
    \begin{align}
    \min\{c(F): F\subseteq E, d_F(B)\ge p\ \forall\ \emptyset\neq B\subsetneq V\}. \label{eq:phase-0}
    \end{align}
This problem admits a polynomial-time $2$-approximation via classic results \cite{Jain2001}. Let $F_0$ be the solution obtained. 

Next, we perform repeated augmentations. If $q=0$, then we return $F$. Otherwise, we execute augmentation phases incrementally for $i=1, 2, \ldots, q$: In the $i$'th augmentation phase, we start with a subset $F_{i-1}$ of edges such that the graph $H_{i-1}:=(V, F_{i-1})$ is \pqminusoneFlexC. We apply the algorithm from Theorem \ref{thm:pqFGCAug} independently $t=\Theta(\log{q})$ times to the $(p, i)$-FGC-Aug instance $(H_{i-1}, p, i, E\setminus F_{i-1}, c|_{E\setminus F_{i-1}})$, where $c|_{E\setminus F_{i-1}})$ is the cost function on $E$ restricted to the edges in $E\setminus F_{i-1}$. Among all $t$ executions that return a feasible augmentation, we pick the cheapest and set it to be $F_i'$; if none of the executions return a feasible augmentation, then we set $F_i'=E\setminus F_{i-1}$. Next we set $F_i:=F_{i-1}\cup F_i'$. Finally, we return $F_q$. 

For each $i=0, 1, \ldots, q$, the graph $(V, F_i)$ is $(p,i)$-Flex-Connected. Thus, the input to phase $i+1$ is a valid instance of \pqFGCAug. Consequently, $F_q$ is \pqFlexC. 

Next, we bound the approximation factor. Let $F^*$ be an optimum solution for the input instance. For phase $0$, we observe that $F^*$ is feasible for \eqref{eq:phase-0} and hence, $c(F)\le 2c(F^*)=2OPT$. 
Consider augmentation phase $i$ for $i\in [q]$. The set $F^*\setminus F_{i-1}$ is a feasible set for the augmentation problem in phase $i$. Hence, the optimum augmentation cost for phase $i$ is at most $OPT$. Each of the $t$ executions has a constant probability of returning an augmentation that is both feasible and has cost at most $O(iL(p, i))OPT$. Hence, the probability that $c(F_i')=\Omega(iL(p, i))OPT$ is $2^{-O(t)}\le 1/100q$ for $t=\Theta(\log{q})$. 
Thus, a union bound over all $q$ phases implies that $c(F)=O(\sum_{i=0}^q i L(p, i))OPT$ with probability at least $1-1/100$. 

For the runtime, we recall that $\alpha(p, q)\le \log{n}$. This in particular, means that $q\le \log{n}$. Therefore, the run-time of all phases $i=1, 2, \ldots, q$ involves $qt=O(q\log{q})=O(\log^2{n})$ calls to the algorithm of Theorem \ref{thm:pqFGCAug} in addition to $qt=O(q\log{q})=O(\log^2{n})$ feasibility testings which is also polynomial-time. Thus, the overall run-time is polynomial in the input size.

\end{proof}

\section*{AI Disclosure}
ChatGPT Astra was used in developing the ideas underlying this work. The authors assume full responsibility for all content. 

\bibliographystyle{abbrv}
\bibliography{references}

\appendix

\end{document}